\documentclass[conference]{IEEEtran}
\IEEEoverridecommandlockouts
\usepackage{cite}

\usepackage{amsfonts,amsmath,amssymb,amsthm,mathrsfs,bm,bbm}
\usepackage{algorithmic}
\usepackage{graphicx}
\usepackage{overpic}
\usepackage{textcomp}
\usepackage{xcolor}
\usepackage{flushend}

\ifCLASSOPTIONcompsoc 
	\usepackage[caption=false,font=normalsize,labelfon t=sf,textfont=sf]{subfig} 
\else 
	\usepackage[caption=false,font=footnotesize]{subfig} \fi 
    
\newcommand{\eqdef}{:=}

\newcommand{\rvec}[1]{\mathbbm{#1}} 		
\newcommand{\rmat}[1]{\mathbbm{#1}} 	
\newcommand{\E}{\mathsf{E}}		
\newcommand{\Var}{\mathsf{Var}}			
\newcommand{\stdset}[1]{\mathbbmss{#1}}	
\newcommand{\set}[1]{\mathcal{#1}}		
\renewcommand{\vec}[1]{\bm{#1}}		
\newcommand{\CN}{\mathcal{CN}}			
\newcommand{\herm}{\mathsf{H}}			
\newcommand{\T}{\mathsf{T}}				

\newcommand{\real}{\stdset{R}}			
\newcommand{\Natural}{\stdset{N}}			
\newcommand{\signal}[1]{\vec{#1}}     
\newcommand{\refeq}[1]{\eqref{#1}}

\newtheorem{proposition}{Proposition}
\newtheorem{assumption}{Assumption}
\newtheorem{remark}{Remark}
\newtheorem{corollary}{Corollary} 

\newtheorem{example}{Example}
    
\begin{document}

\title{Closed-form max-min power control for some cellular and cell-free massive MIMO networks
\thanks{L. Miretti, R. L. G. Cavalcante, and S. Sta\'nczak acknowledge the financial support by the Federal Ministry of Education and Research of Germany in the programme of “Souverän. Digital. Vernetzt.” Joint project 6G-RIC, project identification number: 16KISK020K and 16KISK030.}
}

\author{
\IEEEauthorblockN{Lorenzo Miretti, Renato L. G. Cavalcante, Sławomir Sta\'nczak}
\IEEEauthorblockA{\textit{Fraunhofer Institute for Telecommunications Heinrich-Hertz-Institut} \\
\textit{Technical University of Berlin}\\
Berlin, Germany \\
\{lorenzo.miretti,renato.cavalcante,sławomir.stanczak\}@hhi.fraunhofer.de}
\and
~
\IEEEauthorblockN{Martin Schubert, Ronald B\"ohnke, Wen Xu}
\IEEEauthorblockA{\textit{Huawei}\\
Munich, Germany \\
\{martin.schubert,ronald.boehnke,wen.xu\}@huawei.de}
}

\IEEEspecialpapernotice{(Invited Paper)} 
\maketitle

\begin{abstract} Many common instances of power control problems for cellular and cell-free massive MIMO networks can be interpreted as max-min utility optimization problems involving affine interference mappings and polyhedral constraints. We show that these problems admit a closed-form solution which depends on the spectral radius of known matrices. In contrast, previous solutions in the literature have been indirectly obtained using iterative algorithms based on the bisection method, or on fixed-point iterations. Furthermore, we also show an asymptotically tight bound for the optimal utility, which in turn provides a simple rule of thumb for evaluating whether the network is operating in the noise or interference limited regime. We finally illustrate our results by focusing on classical max-min fair power control for cell-free massive MIMO networks.
\end{abstract}

\begin{IEEEkeywords}
power control, massive MIMO, cell-free, max-min fairness, interference calculus, 6G
\end{IEEEkeywords}

\section{Introduction} 
Weighted max-min utility optimization problems are common in power control tasks in wireless networks because the solutions are known to promote fair allocation of resources among users, and to span the boundary of the achievable performance region \cite{slawomir09,emil2013resource,marzetta16,ngo2017cell,demir2021,donato2021,buzzi2020,chaves2020convergence,yang2017massive,massivemimobook}.
This is not true, for instance, for weighted sum utility maximization problems: except for the particular case of convex performance regions, their solutions may miss some boundary points that are important under a user fairness perspective \cite{das1997closer}, \cite[Remark 1.3]{emil2013resource}. 


Many existing globally optimal solvers are based on a bisection method that addresses a sequence of convex feasibility problems at each iteration \cite{donato2021,ngo2017cell}\cite[Algorithm~1]{massivemimobook}. One of the main limitations of solutions of this type is that the feasibility problems need to be solved with high numerical precision \cite[Sect.~3.4.1]{demir2021}. Therefore, depending on the numerical technique being applied, the overarching bisection algorithm can become too slow in large systems. As a means of addressing this potential drawback,  the algorithm in \cite[Algorithm~3.1]{demir2021} 
considers a sequence of standard optimization problem for which fast algorithms are readily available. Nevertheless, the resulting scheme can still be too complex for real time implementation. To address this issue, recently, some authors even started to consider suboptimal solutions based on meta-heuristics \cite{conceicao2022heuristics}.

Promising optimal alternatives to the bisection method are the iterative power control algorithms in \cite[Algorithm~3.2]{demir2021}, \cite{schubert2019,renatomaxmin,renato2016maxmin,ismayilov2019power}, which can be seen as applications of the mathematical tools reviewed in \cite{nuzman07,renato2019} and the references therein. However, many of these algorithms are tailored to challenging problems going beyond simple power control, 
and they do not fully exploit the affine structure of some power control problems of practical interest. Therefore, simple expressions for the optimal solution have not been obtained. 

More specifically, in some max-min utility optimization problems, which cover the power control problems in cellular and cell-free networks described in \cite[Ch.~5.3.2]{marzetta16}, \cite{ngo2017cell} \cite[Ch.~7]{demir2021} \cite[Theorem~7.1]{massivemimobook},\cite{chaves2020convergence,yang2017massive} as particular instances, the optimal power allocation can be seen as the fixed point of an affine mapping, and the power constraint is a polyhedral set. 
By exploiting this structure and the results in \cite{renato2019}, Section~\ref{sec:power_control} presents a simple expression for a solution to this class of problems. In particular, we show that the limit point (i.e., the optimal power allocation) sought by some existing iterative power control algorithms can be straightforwardly obtained by computing the spectral radius of known matrices and by solving only one system of linear equations that, unlike the approaches in  \cite{chaves2020convergence,yang2017massive}, do not require the bisection method for its construction. 
We also specialize the bounds in \cite{renato2019} to gain insights into the expression of the optimal utility. These bounds are useful to identify whether a  network is operating close to an interference limited or noise limited regime. 

Section~\ref{sec:cellfree} illustrates the practical application of the above results to cell-free massive multiple-input multiple-output (MIMO) networks \cite{demir2021}, by revisiting common instances of uplink (UL) and downlink (DL) power control problems. The main feature of these problems is that their solutions depend only on long-term channel statistics. As a result, we open up the possibility of performing real-time, large-scale, and globally optimal power control, e.g., at a remote central controller, without requiring excessive fronthaul and computational resources.

\subsection{Related results}
Although somewhat unkown within the cellular and cell-free massive MIMO literature, our solution to the max-min utility optimization problem in Section~\ref{sec:power_control} has been rediscovered many times in different contexts. For instance, under slightly more restrictive conditions, a related solution  can be found in \cite[Algorithm~5.3]{slawomir09}. Furthermore, the same solution can also be obtained from \cite[Appendix~A]{cai2012maxmin}, where it was used as a basis for addressing a more general problem. In this work we provide an alternative proof and show that, in fact, the solution follows from a simple application of \cite{nuzman07}. Moreover, from a more practical perspective, we show how this solution can be directly applied to popular power control problems in the modern cellular and cell-free massive MIMO frameworks.

\subsection{Notation and mathematical preliminaries}
We denote by $\real_+$ and $\real_{++}$ the sets of, respectively, nonnegative and positive reals. The spectral radius of a matrix $\signal{M}\in\real^{K\times K}$ is denoted by $\rho(\signal{M})$.  A norm $\|\cdot\|$ on $\real^K$ is said to be \emph{monotone} (on the nonnegative orthant) if $(\forall\signal{x}\in\real_+^K)(\forall\signal{y}\in\real_+^K)~\signal{x}\le\signal{y}\Rightarrow \|\signal{x}\|\le\|\signal{y}\|$, where inequalities involving vectors should be understood coordinate-wise. A mapping $T:\real_{+}^K\to\real^K_{++}$ is called a standard interference mapping if the following properties hold \cite{yates95}: 
\begin{itemize}
	\item[(i)] [monotonicity] $(\forall \signal{x}\in\real^K_{+})(\forall \signal{y}\in\real^K_{+}) ~ \signal{x}\ge\signal{y} \Rightarrow T(\signal{x})\ge T(\signal{y})$
	\item[(ii)] [scalability] $(\forall \signal{x}\in\real^K_+)$ $(\forall \alpha>1)$  $\alpha {T}(\signal{x})>T(\alpha\signal{x})$.
\end{itemize}
For later reference, given a standard interference mapping, we call each coordinate function of the mapping a \emph{standard interference function}. For simplicity, in the remainder of this study we also require continuity for a mapping to be called a standard interference function. A mapping $T:\real_+^K\to\real_{++}^K$ is said to be a \emph{positive concave mapping} if each coordinate function is positive and concave, and we recall that positivity and concavity imply properties (i) and (ii) above  \cite[Proposition~1]{renato2016}. The set of fixed points of $T:\real_+^K\to\real_+^K$ is denoted by $\mathrm{Fix}(T):=\{\signal{x}\in\real_+^K~|~T(\signal{x})=\signal{x}\}$, and we note that standard interference mappings have at most one fixed point \cite{yates95}. 
\begin{proposition}{\cite[Theorem~3.2]{nuzman07}} Given a standard interference mapping $T:\real_+^K\to\real_{++}^K$ and a monotone norm $\|\cdot\|$, the following \emph{conditional eigenvalue} problem is guaranteed to have a unique solution:
\begin{equation}\label{eq:P0}
	\begin{array}{rl}
		\mathrm{find} & (\lambda,\signal{x})\in \real_{++}\times \real_{++}^K \\ 
		\mathrm{such~that}& T(\signal{x})=\lambda \signal{x}, \quad \|\signal{x}\|=1. \\
	\end{array}
\end{equation}
\end{proposition}
We finish this section with a known result that is used to keep the proof of our main contributions self-contained.


\begin{proposition}
	\label{fact.spec_bound} Let $\signal{M}\in\real^{K\times K}_+$ and $(\lambda,\signal{x})\in\real_{++}\times\real_{++}^K$ satisfy the inequality $\signal{Mx}\le\lambda\signal{x}$. Then the spectral radius $\rho(\signal{M})$ of $\signal{M}$ is upper bounded by $\lambda$, i.e., $\rho(\signal{M})\le\lambda$.
\end{proposition}
\begin{proof}
	From the classical Perron-Frobenius theorem \cite[Theorem~2.4.1(i)]{krause2015},  we known that $\signal{M}$ has a maximal nonnegative eigenvalue $\rho(\signal{M})$, and there exists $\signal{v}\in\real^K_+\backslash\{\vec{0}\}$ such that $\signal{Mv}=\rho(\signal{M})\signal{v}$. Now, consider the following linear mapping defined on the nonnegative cone $\real_+^K$:
	\begin{align*} T:\real_+^K\to\real_+^K:\signal{u}\mapsto \signal{Mu}.
		\end{align*}
	Since $T(\signal{x})=\signal{Mx}\le\lambda\signal{x}$ and $\signal{x}\in\real_{++}^K$ by assumption, it follows from \cite[Lemma~3.3]{nussbaum1986convexity} that  $T(\signal{v})=\signal{Mv}=\rho(\signal{M})\signal{v}$ implies $\rho(\signal{M})\le\lambda$, and the proof is complete.
\end{proof}

\section{Max-min power control with polyhedral power constraints}\label{sec:power_control}
\subsection{Problem statement}
As discussed in \cite[Ch.~3, Ch. 7]{demir2021}, many existing uplink and downlink power control algorithms for cellular and cell-free massive MIMO networks solve particular instances of the following optimization problem, which can be seen as a generalization of the power control problems in \cite[Ch.~5.6.3]{slawomir09} \cite[Ch.~5.3.2]{marzetta16} \cite[Theorem~7.1]{massivemimobook}\cite{chaves2020convergence,yang2017massive}:
\begin{equation}\label{eq:P1}
	\begin{array}{rl}
		\underset{\vec{p}=[p_1,\ldots,p_K]^\T\in\stdset{R}^K}{\mathrm{maximize}}&\min_{k\in\{1,\ldots,K\}} \dfrac{b_k~p_k}{\signal{c}_k^\T\signal{p}+\sigma_k} \\ 
		\mathrm{subject~~to}& (\forall~n\in\{1,\ldots,N\})~\signal{a}_n^\T\signal{p}\le p_\mathrm{max}\\
		&\signal{p}\geq \signal{0},
	\end{array}
\end{equation}
where $K\in\Natural$, $N\in\Natural$, $p_\mathrm{max} \in \real_{++}$,
\begin{align*}
[\vec{a}_1,\ldots,\vec{a}_N] &=: \vec{A}\in \real_{+}^{K\times N}, \\
(b_1,\ldots,b_K) &=: \vec{b} \in \real_{++}^{K},\\
[\vec{c}_1,\ldots,\vec{c}_K] &=: \vec{C}\in \real_{+}^{K\times K}, \text{ and} \\
(\sigma_1,\ldots, \sigma_K)&=: \vec{\sigma}\in\real_{++}^K
\end{align*}
are problem parameters. For example, in uplink power control problems, $K$ is the number of users in the network, $\vec{A}$ collects $N$ linear constraints on the power vector $\signal{p}$ (the optimization variable), $p_{\max}$ is the maximum allowed transmit power, $\vec{\sigma}$ is the vector of noise powers, and  $\vec{b}$ and $\vec{C}$ are parameters used to model the effective channel for some \textit{fixed} network configuration (e.g., the choice of beamformers). More detailed examples related to cell-free massive MIMO are provided in Section~\ref{sec:cellfree}; we simply anticipate that the cost function in Problem~\eqref{eq:P1} can be interpreted as the lowest signal-to-interference-plus-noise ratio (SINR) among all users in the network.  We emphasize that, in the above formulation, the parameters $\vec{A}$, $\vec{b}$, $\vec{C}$, and $\vec{\sigma}$ must not depend on $\vec{p}$. 


\subsection{Optimal solution}
The main contribution of this section is to show that a solution to Problem~\eqref{eq:P1} can be obtained by solving a single system of linear equations that can be easily constructed with the problem parameters. The resulting system does not require, for example, the bisection method for its construction, as done in \cite{chaves2020convergence,yang2017massive} to solve a particular instance of Problem~\eqref{eq:P1}. 

To derive this result, which is shown later in Proposition~\ref{proposition.main}, we first need to introduce a (mild) technical assumption. Denote by $S\subset\real_+^K$ the set of power vectors satisfying all constraints in Problem~\eqref{eq:P1}, which we call the \emph{feasible set}. Hereafter, we impose the following natural assumption on $S$, which implicitly restricts the choice of the matrix $\vec{A}$.
\begin{assumption}
	\label{assumption.power}
	The set of feasible power allocations $S$ is bounded, and $S\cap\real_{++}^K\neq\emptyset$.
\end{assumption}

Boundedness of the set $S$ is expected because in any power control problem the transmit power is limited by law or by hardware capabilities. In some power control problems, existence of $\signal{p}\in\real_{++}^K\cap S$  can be intuitively understood as allowing every user to be served by the network. With the above assumption, the set $S$ has the following additional properties:
\begin{remark}
	\label{remark.properties}
	The feasible set $S$ is a compact convex set with nonempty interior. Furthermore, it is downward comprehensive on the nonnegative cone $\real_{+}^K$, i.e., $(\forall \signal{q}\in\real_+^K) (\forall\signal{p}\in S)~\signal{q}\le \signal{p}\Rightarrow \signal{q}\in S$.
\end{remark}
\begin{proof}
	The set $S$ is closed and convex because it is the intersection of the closed convex sets, or, more precisely, the cone $\real_{+}^K$ and the closed half-spaces $(\forall n\in\{1,\ldots,N\}) ~H_n:=\{\signal{p}\in\real^K~|~\signal{a}_n^\T\signal{p}\le p_\text{max}\}$. Being bounded by assumption (see Assumption~\ref{assumption.power}) and closed, $S$ is compact because the space $\real^K$ is finite dimensional. The fact that $S$ is downward comprehensible on $\real_{+}^K$ is immediate from nonnegativity of the vectors $(\signal{a}_n)_{n\in\{1,\ldots,N\}}$ and the definition of the half-spaces $(H_n)_{n\in\{1,\ldots,N\}}$. Downward comprehensibility of $S$ and the existence of a vector $\signal{p}\in\real_{++}^K\cap S$ imply that $S$ has nonempty interior because, for any norm $\|\cdot\|$ on $\real^K$ and for $\epsilon>0$ sufficiently small,\linebreak[4] $ \emptyset\neq B:=\{\signal{q}\in\real^K~|~ \|(1/2)\signal{p} - \signal{q}\|<\epsilon\}\subset S.$ 
\end{proof}

\begin{remark}
	\label{remark.norm}
Remark~\ref{remark.properties} and  \cite[Proposition~2]{renatomaxmin} imply that there exists a monotone norm $\|\cdot\|_\star$ on $\stdset{R}^K$ such that the feasible set $S$ can be equivalently written as $S=\{\signal{p}\in \real_+^K ~ | ~ \|\signal{p}\|_\star \le 1\}$. For the particular setting in Problem~\eqref{eq:P1}, this monotone norm $\|\cdot\|_\star$  takes the form
\begin{align}
	\label{eq.norm}
	(\forall \signal{p}\in\real^K)~\|\signal{p}\|_\star := \dfrac{1}{p_{\max}}  \max_{n\in\{1,\ldots,N\}} \signal{a}_n^\T |\signal{p}|,
\end{align}
where $|\signal{p}|$ is the column vector obtained by taking the coordinate-wise absolute value of a vector $\signal{p}$, i.e.,  $(\forall\signal{p}=[p_1,\ldots,p_K]\in\real_+^K)~ |\signal{p}|:=[|p_1|,\dots,|p_K|]^\T$.
\end{remark}
Below we further simplify $\|\cdot\|_\star$ in Remark~\ref{remark.norm} for two common choices of $\vec{A}$: 
\begin{example}
	If $K=N$ and $\vec{A} = \vec{I}_K$ is the $K$-dimensional identity matrix, then $\|\cdot\|_\star$ reduces to a scaled  $l_\infty$-norm: 
\begin{equation*}
		\|\signal{p}\|_\star=\dfrac{1}{p_\mathrm{max}} \max_{k\in\{1,\ldots,K\}} |p_k| = \dfrac{1}{p_\mathrm{max}}\|\vec{p}\|_\infty.
\end{equation*}
\end{example}

\begin{example}
	If $N=1$ and $\vec{A} = \vec{1}_{K\times 1}$ is the $K$-dimensional vector of ones, then $\|\cdot\|_\star$ reduces to a scaled $l_1$-norm: 
	\begin{align*}
		\|\signal{p}\|_\star = \dfrac{1}{p_\mathrm{max}} \sum_{k\in\{1,\ldots,K\}} |p_k| =\dfrac{1}{p_\mathrm{max}}\|\vec{p}\|_1.
	\end{align*}
\end{example}

Now, consider the following optimization problem:
\begin{equation}\label{eq:P2}
	\begin{array}{rl}
		\underset{(t,\vec{p})\in\real_+\times\real^K_{++}}{\mathrm{maximize}} &t \\ 
		\mathrm{subject~~to}& \signal{p}=t(\signal{Mp}+\signal{u}), \quad \|\signal{p}\|_\star \le 1.
	\end{array}
\end{equation}
where $\|\cdot\|_\star$ is the monotone norm in \eqref{eq.norm} induced by the power constraints $(\vec{A},p_{\max})$ satisfying Assumption \ref{assumption.power}, and where we define the scaled coefficients
\begin{align*}
\vec{M} &:= \mathrm{diag}(\vec{b})^{-1}\vec{C}^\T,\\
\vec{u} &:= \mathrm{diag}(\vec{b})^{-1}\vec{\sigma}.
\end{align*}
Problems~\eqref{eq:P1} and \eqref{eq:P2} are related in the next proposition, which can be proved using standard arguments in the literature (see, for example, the discussion below \cite[Eq.~(3.34)]{demir2021} and the references therein), so we omit the proof for brevity.   
\begin{proposition}
	\label{proposition.epigraph} If the tuple $(t^\star,\signal{p}^\star)\in\real_{++}\times \real_{++}^K$ solves Problem~\eqref{eq:P2}, then $\signal{p}^\star=:[p_1^\star,\ldots,p_K^\star]^\T$ is a solution to Problem~\eqref{eq:P1}, and the optimal value of \eqref{eq:P1} is \begin{equation}
	\min_{k\in\{1,\ldots,K\}} \dfrac{b_k~p_k^\star}{\signal{c}_k^\T\signal{p}^\star+\sigma_k}=t^\star.
	\end{equation}
\end{proposition}

In light of Proposition~\ref{proposition.epigraph}, we can obtain a solution to Problem~\eqref{eq:P1} by solving  Problem~\eqref{eq:P2}, which has a unique solution with the simple expression derived below.

\begin{proposition}
	\label{proposition.main}
Problem~\eqref{eq:P2} has a unique solution $(t^\star,\signal{p}^\star)\in\real_{++}\times \real_{++}^K$ given by
	\begin{align}
		\label{eq.topt}
		t^\star=\dfrac{1}{\underset{n\in\{1,\ldots,N\}}{\max}\rho(\signal{M}_n)}
	\end{align} 
	and
	\begin{align}
		\label{eq.optp}
		\signal{p}^\star=t^\star(\signal{I}-t^\star\signal{M})^{-1}\signal{u},
	\end{align}
where
	\begin{equation}\label{eq:Mn}
	(\forall n\in\{1,\ldots,N\})~	
		\signal{M}_n:=\signal{M}+\dfrac{1}{p_\mathrm{max}}\signal{u}\signal{a}_n^\T.
	\end{equation}
Furthermore, $\vec{p}^\star$ is also the eigenvector with norm $\|\signal{p}^\star\|_\star=1$ associated with the largest eigenvalue of any matrix $\vec{M}_n$ in \eqref{eq:Mn} satisfying $\rho(\vec{M}_n) = 1/{t^\star}$ for some $n\in\{1,\ldots,N\}$. 
\end{proposition}
\begin{proof}
	It follows from Remark~\ref{remark.norm} and the results in \cite{nuzman07} (see also the discussion in \cite{renato2019}) that $(t^\star,\signal{p}^\star)\in\real_+\times\real_{++}^K$ solves Problem~\eqref{eq:P2} if and only if $(1/t^\star,\signal{p}^\star)$ is the unique solution to the condition eigenvalue problem~\eqref{eq:P0} with the  monotone norm $\|\cdot\|_\star$ in \refeq{eq.norm} and the mapping 
	\begin{align}
		\label{eq.mapping}
		T:\real_{+}^K \to \real_{++}^K:\signal{p}\mapsto \signal{Mp}+\signal{u}.
	\end{align}
	As a result, from the definition of problem~\eqref{eq:P0}, we verify that
\begin{equation}
		\label{eq.pstar_prop}
		\signal{p}^\star = t^\star (\signal{Mp}^\star+\signal{u}), \quad \|\signal{p}^\star\|_\star=1.
\end{equation}
Now let $i\in\{1,\ldots,N\}$ be any index satisfying
	\begin{equation}
		\label{eq.ineq_norm}
		\|\signal{p}^\star\|_\star = \frac{1}{p_{\max}} \signal{a}_{i}^\T\signal{p}^\star \ge \frac{1}{p_{\max}} \mathrm{max}_{n\in\{1,\ldots,N\}} \signal{a}_{n}^\T\signal{p}^\star.
\end{equation}
We then have $(\forall n\in\{1,\ldots,N\})$
\begin{equation}\label{eq.matineq}
	\begin{split}
		 \signal{M}_n\signal{p}^\star &= \signal{Mp}^\star+~\dfrac{1}{p_\mathrm{max}}\signal{u}\signal{a}_n^\T\signal{p}^\star \\
		&\overset{\text{(a)}}{\leq} \signal{Mp}^\star+~\dfrac{1}{p_\mathrm{max}}\signal{u}\signal{a}_i^\T\signal{p}^\star \\
		&= \signal{M}_i \signal{p}^\star \overset{\text{(b)}}{=} \dfrac{1}{t^\star} \signal{p}^\star,
	\end{split}
\end{equation}
where (a) follows from the inequality in \refeq{eq.ineq_norm} and positivity of $\signal{u}$, and (b) follows from \refeq{eq.pstar_prop} and the equality in \eqref{eq.ineq_norm}. Consequently, the equality (b) in \refeq{eq.matineq} shows that the solution $(t^\star,\signal{p}^\star)$ to problem~\eqref{eq:P2} can be obtained from an eigenpair $(1/t^\star,\vec{p}^\star)\in\real_{++}\times\real^K_{++}$ of the matrix $\signal{M}_i$ with $\|\signal{p}^\star\|_\star=1$. Furthermore, in light of Proposition~\ref{fact.spec_bound}, we verify from the equality (b) in \refeq{eq.matineq} that $1/t^\star>0$ is the spectral radius of $\signal{M}_i$ (i.e., $t^\star = 1/\rho(\signal{M}_i)$) and from the inequality (a) in \refeq{eq.matineq} that 
	\begin{align}
		\label{eq.boundspec}
		(\forall n\in\{1,\ldots,N\})~ \rho(\signal{M}_n) \le \rho( \signal{M}_i).
	\end{align}
This completes the proof of \refeq{eq.topt} and the statement in the last sentence of the proposition. The equality in \refeq{eq.optp} follows directly from the definition of the equality constraint of Problem~\eqref{eq:P2}, and the proof is complete. (Note that the matrix $(\signal{I}-t^\star\signal{M})$ is full rank because otherwise there would exist infinitely many vectors $\signal{p}\in\real_{++}^K$ satisfying $\signal{p}=t^\star(\signal{Mp}+\signal{u})$, thus contradicting uniqueness of the solution to Problem~\eqref{eq:P2}, which has already been proved above.)
\end{proof}

\subsection{Bound on the optimal utility}
The next corollary, which is immediate from \cite[Proposition 3]{renato2019}, shows a simple bound for the optimal utility $t^\star$ in \refeq{eq.topt} as a function of $p_\mathrm{max}$:
\begin{corollary}
	\label{cor.maxminbound}
Assume that $\rho(\signal{M})>0$, and denote by $(t^\star(p_\mathrm{max}), \signal{p}^\star(p_\mathrm{max}))$ the solution in \refeq{eq.topt} and \refeq{eq.optp} to an instance of Problem \eqref{eq:P2}, where we made explicit the dependency of this solution on the choice  of the maximum transmit power $p_\mathrm{max}>0$. Then $(\forall p_\mathrm{max}\in\real_{++})$ 
\begin{align}\label{eq.tmax} \quad t^\star(p_\mathrm{max})~ \le \begin{cases} 1/\rho(\signal{M})&\text{if } p_\mathrm{max} \ge p_\mathrm{T},\\
			p_\mathrm{max}/\|\signal{u}\|&\text{otherwise},
		\end{cases}
	\end{align}
	where $p_\mathrm{T}$ is the transition point defined by $p_\mathrm{T}:=\|\signal{u}\|/\rho(\signal{M})$, and $\|\cdot\|$ is the monotone norm 
	\begin{align*}
		(\forall \signal{p}\in\real^K)~\|\signal{p}\| :=  \max_{n\in\{1,\ldots,N\}} \signal{a}_n^\T |\signal{p}|.
	\end{align*}
	Furthermore, the bound in \refeq{eq.tmax} is asymptotically tight as $p_\mathrm{max}\to 0^+$ and as $p_\mathrm{max}\to \infty$.
\end{corollary}

%

\begin{remark}
	\label{remark.transition}
	As discussed in \cite{renato2019}, in a power control problem, as the maximum transmit power $p_\mathrm{max}$ increases, the network configuration obtained by solving Problem~\eqref{eq:P1} moves from a noise limited regime to an interference limited regime around the transition point $p_\mathrm{T}$ defined in Corollary~\ref{cor.maxminbound}.
\end{remark}

\section{Applications to cell-free networks}
\label{sec:cellfree}
\subsection{System model}
We now apply the above results to a cell-free massive MIMO network, as defined in \cite{demir2021}, composed of $L$ access-points (APs) indexed by $\mathcal{L}:=\{1,\ldots,L\}$, each of them equipped with $M$ antennas, and $K$ single-antenna user equipments (UEs) indexed by $\mathcal{K}:=\{1,\ldots,K\}$. The UL and DL ergodic rates simultaneously achievable by each UE can be lower bounded by \cite{marzetta16,massivemimobook} $(\forall k\in\set{K})$
\begin{align*}
R_k^{\mathrm{UL}}(\vec{p}) &\eqdef \log_2(1+\mathrm{SINR}_k^{\mathrm{UL}}(\vec{p})) \quad \text{[bit/s/Hz]}, \\
R_k^{\mathrm{DL}}(\vec{p}) &\eqdef \log_2(1+\mathrm{SINR}_k^{\mathrm{DL}}(\vec{p})) \quad \text{[bit/s/Hz]},
\end{align*}
\begin{align*}
\mathrm{SINR}_k^{\mathrm{UL}}(\vec{p}) &\eqdef \dfrac{p_k|\E[\rvec{h}_k^\herm\rvec{v}_k]|^2}{p_k\Var(\rvec{h}_k^\herm\rvec{v}_k)+\sum_{j\neq k}p_j\E[|\rvec{h}_j^\herm\rvec{v}_k|^2]+\sigma}, \\
\mathrm{SINR}_k^{\mathrm{DL}}(\vec{p}) &\eqdef \dfrac{p_k|\E[\rvec{h}_k^\herm\rvec{v}_k]|^2}{p_k\Var(\rvec{h}_k^\herm\rvec{v}_k)+\sum_{j\neq k}p_j\E[|\rvec{h}_k^\herm\rvec{v}_j|^2]+\sigma},
\end{align*}
where $\rmat{h}_{k}$ is a \textit{random} vector taking values in $\stdset{C}^{ML}$, which models the fading channel between all APs and UE~$k$; $\rvec{v}_k$ is a \textit{random} vector taking values in $\stdset{C}^{ML}$, which satisfies without loss of generality $\E[\|\rvec{v}_k\|^2]=1$ {(here, $\|\cdot\|$ denotes the standard $l^2$ norm), and it is applied by all APs to filter jointly the signal of UE~$k$; $\sigma>0$ is the noise variance; and $\vec{p} = [p_1,\ldots,p_K]^\T \geq \vec{0}$ is a vector of \textit{deterministic} (i.e., long term) power scaling coefficients. 

In the context of cell-free massive MIMO, $p_k$ should be interpreted as the transmit power of UE $k$ for the UL case, or as the \textit{total} transmit power used by the network to serve UE $k$ for the DL case. Similarly, $\rvec{v}_k$ should be interpreted as a joint UL \textit{combiner} for the message of UE $k$, or as a joint DL \textit{precoder} for the message of UE $k$. In practice, the filters $\{\rvec{v}_k\}_{k\in\set{K}}$ need to satisfy additional constraints modeling limited APs cooperation, for instance, w.r.t. channel state information sharing (CSI) and joint encoding/decoding capabilities. Popular examples are the \textit{distributed} or \textit{centralized} \textit{user-centric} clustered models reviewed in \cite{demir2021}, or the special case of \textit{cellular} massive MIMO \cite{massivemimobook}. Since the focus of this article is on power control, we omit the details on such constraints, and refer to \cite{miretti2021team,miretti2021team2} for details. In the following, we simply assume that $\{\rvec{v}_k\}_{k\in\set{K}}$ are given, and that the expectations in the following matrices and vector exist: 
\begin{equation*}
\vec{G} \eqdef \begin{bmatrix}
\E[|\rvec{h}_1^\herm\rvec{v}_1|^2] & \ldots & \E[|\rvec{h}_1^\herm\rvec{v}_K|^2] \\
\vdots & \ddots & \vdots \\
\E[|\rvec{h}_K^\herm\rvec{v}_1|^2] & \ldots & \E[|\rvec{h}_K^\herm\rvec{v}_K|^2] 
\end{bmatrix},
\end{equation*}
\begin{equation*}
\vec{d} \eqdef \left(|\E[\rvec{h}_1^\herm\rvec{v}_1]|^2, \ldots, |\E[\rvec{h}_K^\herm\rvec{v}_K]|^2\right), \quad \vec{D} \eqdef \mathrm{diag}(\vec{d}).
\end{equation*}

\subsection{Uplink power control with per-UE power constraint}\label{ssec:UL}
By fixing an arbitrary combining design, and a per-UE power constraint $(\forall k\in\set{K})$ $p_k\leq p_{\max}$, each positive point on the boundary of the UL rate region achieved by power control can be obtained as the solution to problems of the type
\begin{equation}\label{prob.UL}
\underset{\vec{p}\geq \vec{0},~\|\vec{p}\|_{\infty}\leq p_{\max}}{\mathrm{maximize}}\min_{k\in\set{K}}\omega_k^{-1}\mathrm{SINR}^{\mathrm{UL}}_k(\vec{p})
\end{equation}
for some vector of weights $\bm{\omega}\eqdef(\omega_1,\ldots,\omega_K)\in\stdset{R}_{++}^K$. The solution to Problem~\eqref{prob.UL} can be computed in closed form following the methodology in Section~\ref{sec:power_control}. Specifically, we can map Problem~\eqref{prob.UL} to Problem~\eqref{eq:P1}, and apply Proposition~\ref{proposition.epigraph} and Proposition~\ref{proposition.main}, after identifying
\begin{equation*}
\vec{A}=\vec{I}_K, \quad \vec{b} = \mathrm{diag}(\vec{\omega})^{-1}\vec{d},\quad \vec{C} = \vec{G}-\vec{D}, \quad \vec{\sigma} = \sigma\vec{1}.
\end{equation*}

\begin{remark} An alternative way of finding boundary points on the rate region achieved by power control is to consider problems of the type 
\begin{equation}\label{prob.rate.UL}
\underset{\vec{p}\geq \vec{0},~\|\vec{p}\|_{\infty}\leq p_{\max}}{\mathrm{maximize}}\min_{k\in\set{K}}\omega_k^{-1}R^{\mathrm{UL}}_k(\vec{p}).
\end{equation}
In general, using the same weights $\vec{\omega}$ in \eqref{prob.rate.UL} and \eqref{prob.UL} produces different boundary points. A notable exception is the choice $\vec{\omega}=\vec{1}$, which produces in both cases the so-called \emph{max-min fair} point. This boundary point is popular in the cell-free literature, where the main motivation is indeed to improve user fairness w.r.t. cellular networks \cite{ngo2017cell}.
\end{remark}

\subsection{Downlink power control with sum power constraint}
Similarly to the UL case, by fixing an arbitrary precoding design and a sum power constraint $\sum_{k\in\set{K}}p_k\leq p_{\max}$, each positive point on the boundary of the DL rate region achieved by power control can be obtained as the solution to problems of the type
\begin{equation}\label{prob.DL}
\underset{\vec{p}\geq \vec{0},~\|\vec{p}\|_1\leq p_{\max}}{\mathrm{maximize}}\min_{k\in\set{K}}\omega_k^{-1}\mathrm{SINR}^{\mathrm{DL}}_k(\vec{p}).
\end{equation}
for some vector of weights $\bm{\omega}\eqdef(\omega_1,\ldots,\omega_K)\in\stdset{R}_{++}^K$. The solution to Problem~\eqref{prob.DL} can be similarly computed by applying Proposition~\ref{proposition.epigraph} and Proposition~\ref{proposition.main} after identifying
\begin{equation*}
\vec{A}=\vec{1}_{K\times 1}, \quad \vec{b} = \mathrm{diag}(\vec{\omega})^{-1}\vec{d},\quad \vec{C} = \vec{G}^\T-\vec{D}, \quad \vec{\sigma} = \sigma\vec{1}.
\end{equation*}
The main difference w.r.t. the UL case is the transpose operator in the channel gain matrix $\vec{G}$, and the fact that using a single linear power constraint removes the $\max$ operator in the expression of the optimal solution \eqref{eq.topt}.
\begin{remark} The considered DL sum power constraint refers to the power radiated by all APs in the network. Although this is still an important metric for network design/management, in practice, due to hardware and regulatory constraints, each AP (or even antenna) is also typically subject to an individual power constraint. We remark that imposing a per-AP power constraint is different from constraining each $p_k$ as in the UL case. However, solving instances of Problem~\eqref{eq:P1} can still be useful for obtaining heuristic solutions, for instance, based on a suitable design of $\{\rvec{v}_k\}_{k\in\set{K}}$ and $(\vec{A},p_{\max})$ such that the  feasibility under a per-AP power constraint is guaranteed. 
\end{remark}

\begin{remark}
For the special case of a cellular network, where each AP serves a disjoint set of UEs, a per-AP power constraint in the DL can be mapped to $L$ linear constraints $\{\vec{a}_l\}_{l\in\set{L}}$, each acting on a separate subvector of $\vec{p}$ corresponding to the UEs served by AP $l$. Therefore, the considered power control framework covers this case.
\end{remark}

\subsection{Numerical examples}
\begin{figure}
\centering
\includegraphics[width=1\columnwidth]{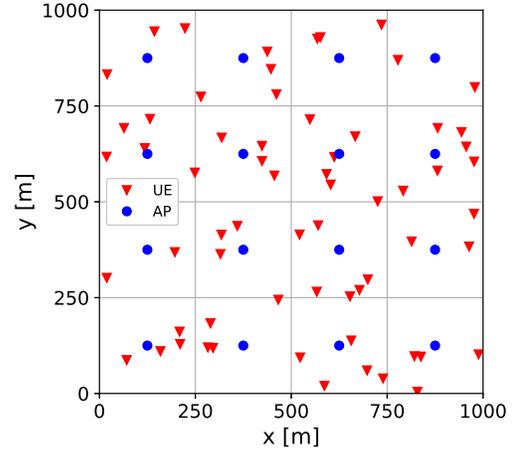}
\caption{Pictorial representation of the simulated setup: $K=64$ UEs uniformly distributed within a squared service area of size $1\times 1~\text{km}^2$, and $L=16$ regularly spaced APs with $M=8$ antennas each. In the cell-free setup, each UE is jointly served by a cluster of $Q=4$ APs offering the strongest channel gains. The cellular case is obtained as a special case by letting $Q=1$.  }
\label{fig:network}
\end{figure}
\begin{figure}
\centering
\includegraphics[width=0.8\columnwidth]{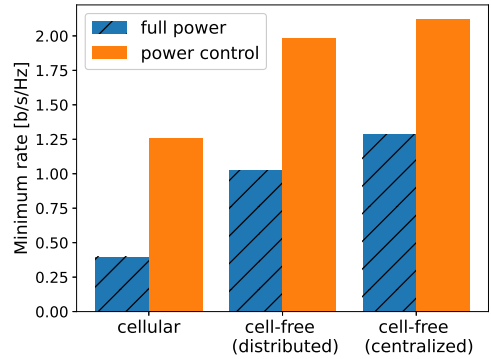}
\caption{Comparison of optimal max-min fair UL power control versus full power transmission for a UE power budget $p_{\max}=20$ dBm, under different AP cooperation regimes. Power control significantly boosts user fairness, especially for regimes with lower interference suppression capabilities.} 
\label{fig:comparison}
\end{figure}
\begin{figure*}[!t]
\centering
\subfloat[]{\includegraphics[width=0.31\textwidth]{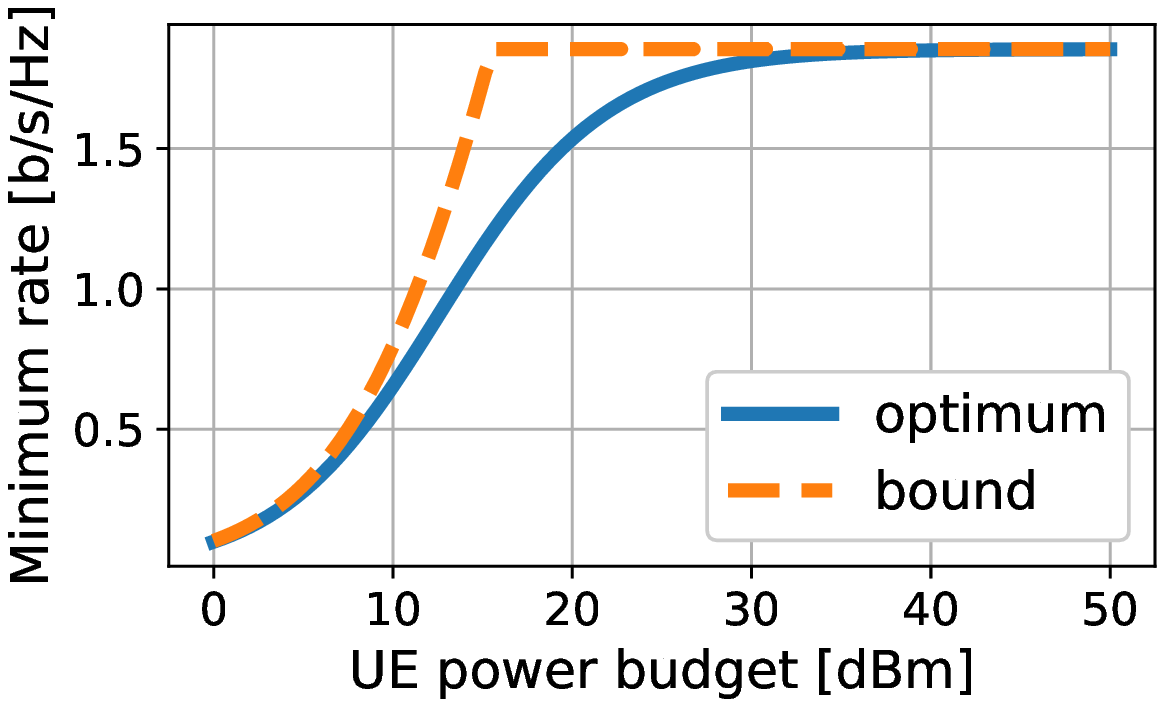}
\label{fig:cell}}
\hfil
\subfloat[]{\includegraphics[width=0.31\textwidth]{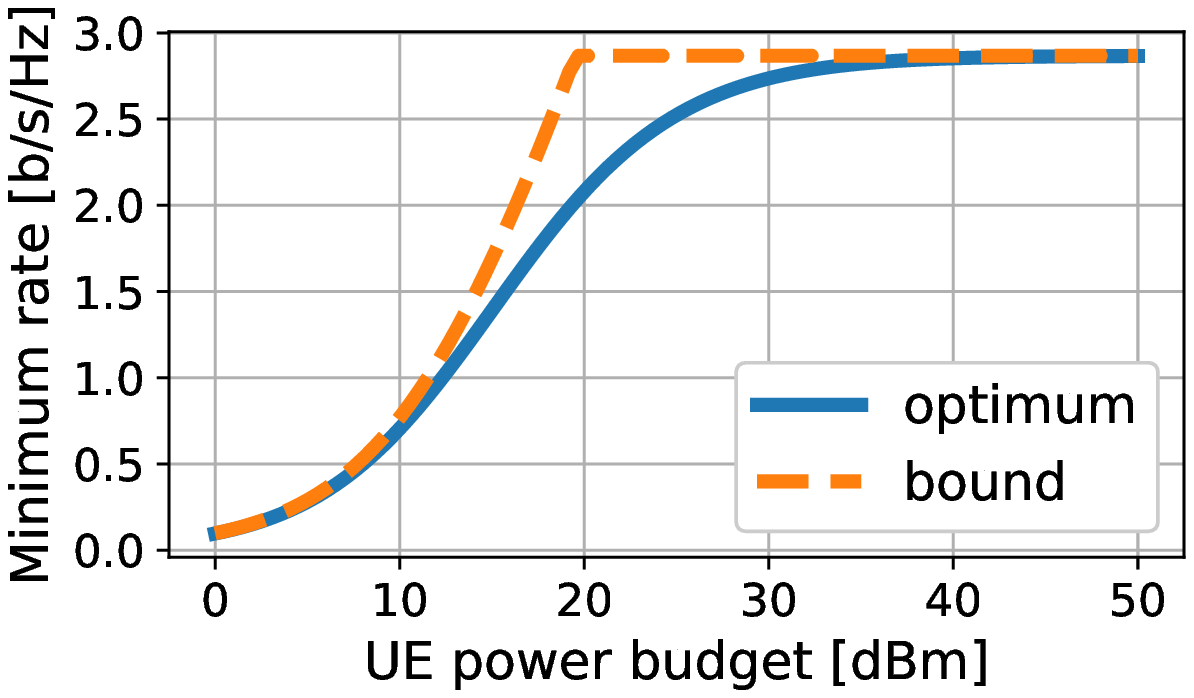}
\label{fig:distr}}
\hfil
\subfloat[]{\includegraphics[width=0.31\textwidth]{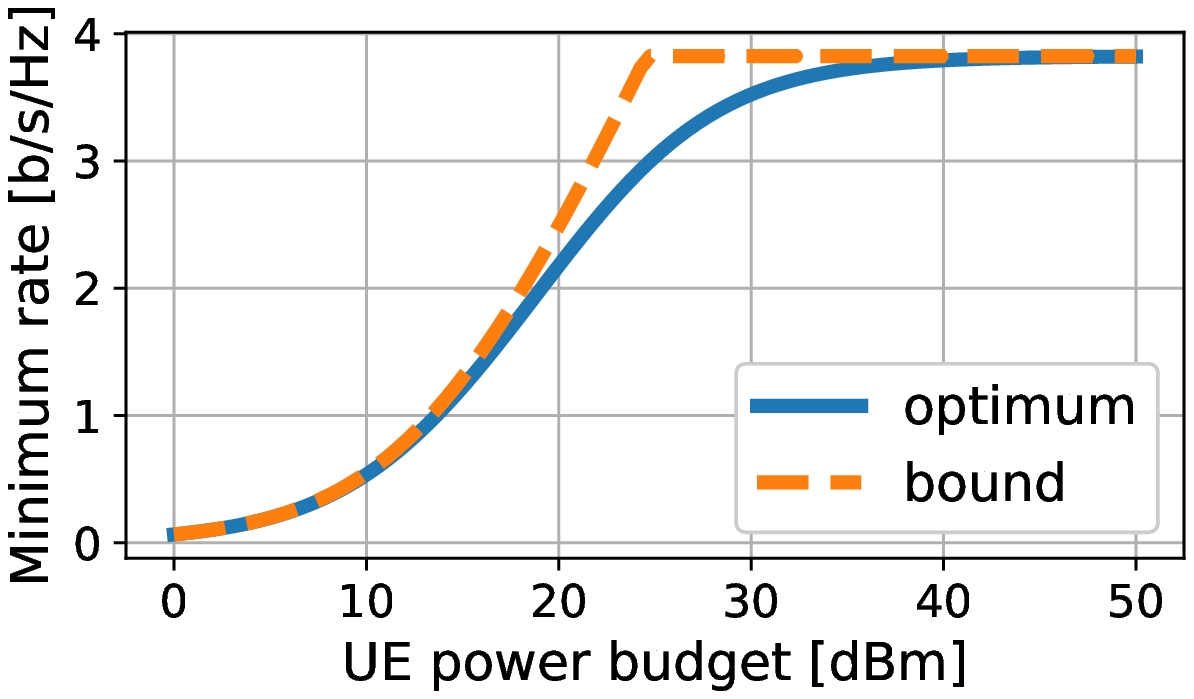}
\label{fig:centr}}
\caption{Optimal max-min fair UL power control for different UE power budgets $p_{\max}$, and related upper bound computed using \eqref{eq.tmax}, for the case of (a) cellular network; (b) distributed cell-free network; (c) centralized cell-free network. By computing the spectral radius of a single known matrix independent from $p_{\max}$, the proposed bound provides accurate performance estimates for the noise limited $(p_{\max}\to 0^+)$ and interference limited $(p_{\max}\to \infty)$ regimes. Furthermore, it also offers a useful rule-of-thumb for estimating the transition point between the two regimes.}
\label{fig:bound}
\end{figure*}
We consider the network depicted in Figure~\ref{fig:network}, where $K=64$ UEs are uniformly distributed within a squared service area of size $1\times 1~\text{km}^2$, and $L=16$ regularly spaced APs with $M=8$ antennas each. By neglecting for simplicity channel correlation, we let each sub-vector $\rvec{h}_{l,k}$ of $\rvec{h}_k^{\herm} =: [\rvec{h}_{1,k}^\herm, \ldots \rvec{h}_{L,k}^\herm]$ be independently distributed as $\rvec{h}_{l,k} \sim \CN\left(\vec{0}, \gamma_{l,k}\vec{I}_M\right)$, where $\gamma_{l,k}>0$ denotes the channel gain between AP $l$ and UE $k$. We follow the same 3GPP-like path-loss model adopted in \cite{demir2021} for a $2$ GHz carrier frequency:
\begin{equation*}
\gamma_{l,k} = -21.9 \log_{10}\left(\dfrac{D_{l,k}}{1 \; \mathrm{m}}\right) -30.5 + Z_{l,k} \quad \text{[dB]},
\end{equation*}
where $D_{l,k}$ is the distance between AP $l$ and UE $k$ including a difference in height of $10$ m, and $Z_{l,k}\sim \CN(0,\rho^2)$ [dB] are shadow fading terms with standard deviation $\rho = 4$. The shadow fading is correlated as $\E[Z_{l,k}Z_{j,i}]=\rho^22^{-\frac{\delta_{k,i}}{9 \text{ [m]}}}$ for all $l=j$ and zero otherwise, where $\delta_{k,i}$ is the distance between UE $k$ and UE $i$. The noise power is 
$\sigma = -174 + 10 \log_{10}(B) + F$ [dBm],
where $B = 20$ MHz is the system bandwidth, and $F = 7$ dB is the noise figure.  

We focus on UL power control as described in Section~\ref{ssec:UL}. For the choice of $\{\rvec{v}_k\}_{k\in\set{K}}$, we consider optimal combiners maximizing $R_k^{\mathrm{UL}}(p_{\max}\vec{1})$ under full power transmission\footnote{We recall that the proposed power control framework requires the combiners to be independent from the free variables $\vec{p}$.} $p_{\max} = 20$ dBm, and under different APs cooperation constraints. In particular, we assume that each UE $k$ is served only by its $Q$ strongest APs, that is, by the the subset of APs indexed by $\set{L}_k\subseteq \set{L}$, where each set $\set{L}_k$ is formed by ordering $\set{L}$ w.r.t. decreasing $\gamma_{l,k}$ and by keeping only the first $Q$ elements. Furthermore, we assume each AP $l$ to acquire (and possibly exchange) local CSI $\hat{\rvec{H}}_l\eqdef[\hat{\rvec{h}}_{l,1},\ldots,\hat{\rvec{h}}_{l,K}]$,  
\begin{equation*}
(\forall k \in\set{K})~\hat{\rvec{h}}_{l,k} \eqdef \begin{cases}
\rvec{h}_{l,k} & \text{if } l \in \set{L}_k,\\
\E[\rvec{h}_{l,k}] & \text{otherwise}.
\end{cases}
\end{equation*}
This model reflects the canonical cell-free massive MIMO implementation with pilot-based local channel estimation, and possible CSI sharing through the fronthaul; we neglect for simplicity estimation/quantization noise, and we focus on a simple model where small-scale fading coefficients are either perfectly known at some APs or completely unknown. In contrast, we assume that the relevant long-term statistical information is perfectly shared within the network. We then study the following AP cooperation regimes\footnote{To avoid cumbersome notation, the combiners described below are given in a unnormalized form. In the following simulations, a deterministic scaling factor is applied to each $\{\rvec{v}_k\}_{k\in\set{K}}$ s.t. $(\forall k \in\set{K})$ $\E[\|\rvec{v}_k\|^2]=1$ holds.}

\subsubsection{Cellular network} We assume each UE to be served only by its strongest AP, i.e., $Q=1$, and no CSI sharing. The $l$-th submatrix of the optimal combining matrix $\rvec{V}\eqdef[\rvec{v}_1,\ldots,\rvec{v}_K]$ corresponding to AP $l$ is the known local MMSE solution \cite{massivemimobook}
\begin{equation*}
(\forall l\in \set{L})~\rvec{V}^{\mathrm{LMMSE}}_l \eqdef \left(\hat{\rvec{H}}_l\hat{\rvec{H}}_l^{\herm}+\vec{\Sigma}_l + \frac{\sigma}{p_{\max}}\vec{I}_M\right)^{-1}\hat{\rvec{H}}_l,
\end{equation*}
where $\vec{\Sigma}_l \eqdef \underset{i \in\{k\in\set{K} | l \notin \set{L}_k\}}{\sum}\gamma_{l,i}\vec{I}_M$ is the CSI error covariance.
\subsubsection{Distributed cell-free network} We assume each UE to be jointly served by its $Q=4$ strongest APs, and no CSI sharing. The $l$-th submatrix of the optimal combining matrix $\rvec{V}\eqdef[\rvec{v}_1,\ldots,\rvec{v}_K]$ corresponding to AP $l$ is the so-called local \textit{team} MMSE solution \cite{miretti2021team,miretti2021team2}:
\begin{equation*}
(\forall l\in \set{L})~\rvec{V}^{\mathrm{LTMMSE}}_l \eqdef \rvec{V}^{\mathrm{LMMSE}}_l\vec{W}_l,
\end{equation*}
where the columns of $\vec{W}_l=:[\vec{w}_{l,1},\ldots,\vec{w}_{l,K}]$ can be computed from the knowledge of $\vec{\Pi}_l\eqdef \E[\hat{\rvec{H}}_l^\herm\rvec{V}^{\mathrm{LMMSE}}_l]$ as the unique solution to the linear system of equations
\begin{equation*}
\begin{cases}
\vec{w}_{l,k} +\sum_{j\neq l}\vec{\Pi}_j\vec{w}_{j,k} = \vec{e}_k & \text{if } l\in\set{L}_k, \\
\vec{w}_{l,k} = \vec{0} & \text{otherwise}.
\end{cases}
\end{equation*}
\subsubsection{Centralized cell-free network} We assume each UE to be jointly served by its $Q=4$ strongest APs, perfectly sharing their CSI. For all $k \in\set{K}$, the subvector $\rvec{v}^{(k)}$ of the optimal joint combiner $\rvec{v}_k$ corresponding to the $Q$ APs serving UE $k$ is the $k$th column of the known MMSE solution \cite{demir2021} 
\begin{equation*}
\rvec{V}^{\mathrm{MMSE},(k)} \eqdef \left(\hat{\rvec{H}}^{(k)}(\hat{\rvec{H}}^{(k)})^{\herm}+\vec{\Sigma}^{(k)} + \frac{\sigma}{p_{\max}}\vec{I}_{QM}\right)^{-1}\hat{\rvec{H}}^{(k)}
\end{equation*}
where $\hat{\rvec{H}}^{(k)}$ is obtained from the global CSI matrix $\hat{\rvec{H}}^\herm \eqdef [\hat{\rvec{H}}_1^\herm, \ldots, \hat{\rvec{H}}_L^\herm]$ by removing from $\hat{\rvec{H}}$ the rows corresponding to the channels of all APs $\l\notin \set{L}_k$, and $\vec{\Sigma}^{(k)}$ is the corresponding CSI error covariance. 

Figure~\ref{fig:comparison} compares the performance of the considered cooperation regimes under optimal max-min fair power control, that is, by letting $\vec{\omega} = \vec{1}$ in Problem~\eqref{prob.UL}, for a UE power budget $p_{\max} = 20$ dBm. As a baseline, we also show the achieved utility (i.e., the minimum rate) by assuming full power transmission $\vec{p} = p_{\max}\vec{1}$. The results are averaged over $N_{\text{setups}}=100$ i.i.d. realizations of the UE positions. In line with related literature (see, e.g., \cite{bjornson2019making}), Figure~\ref{fig:comparison} highlights the importance of optimal power control, especially for tighter AP cooperation constraints impairing the spatial interference suppression capabilities.

Figure~\ref{fig:bound} illustrates the effectiveness of the proposed upper bound \eqref{eq.tmax} in predicting the performance of optimal max-min fair power control under arbitrary UE power budget $p_{\max}$. We remark that the bound depends only on fixed parameters $(\vec{G},\vec{d})$, i.e., on the choice of the combiners and on channel statistics; here, we use the same combiners as in Figure~\ref{fig:comparison}, i.e., the optimal combiners for $(\forall k \in\set{K})$ $p_k = 20$ dBm, and we focus on a single realization of the UE positions. The non-differentiable point in the upper bound curve corresponds to the point $p_T$ after which the system smoothly transitions from the noise limited to the interference limited regime. As predicted by theory, the proposed bound becomes especially tight as the system approaches these two extreme regimes. Furthermore, as intuitively expected, the transition point $p_T$ is smaller for tighter AP cooperation constraints. 

\section{Conclusion}
We have presented a simple expression for the solution to a common class of max-min power control problems, and illustrated its application to cellular and cell-free massive MIMO networks. In previous studies this solution has been characterized as the solution to a conditional eigenvalue problem, as the limit of sequences produced via very specific iterative methods, or as the solution to a system of linear equations that requires the bisection method for its construction. Having a novel characterization of the solution, we now open up the possibility of devising new scalable power control techniques for large-scale problems. In particular, we showed that it suffices to implement standard numerical routines for computing the largest eigenvalue (and the corresponding eigenvector) among a set of matrices constructed from the problem parameters. We also obtained a simple bound for the optimal utility, and simulations have shown that the bound is tight asymptotically, as predicted by theory. In addition, for a large interval of the maximum power constraint, this bound may be also close to the optimum.

\bibliographystyle{IEEEbib}
\bibliography{IEEEabrv,refs}

\begin{thebibliography}{10}

\bibitem{slawomir09}
S.~Sta\'nczak, M.~Wiczanowski, and H.~Boche,
\newblock {\em Fundamentals of Resource Allocation in Wireless Networks},
\newblock Foundations in Signal Processing, Communications and Networking.
  Springer, Berlin Heidelberg, 2nd edition, 2009.

\bibitem{emil2013resource}
E.~Bj{\"o}rnson and E.~Jorswieck,
\newblock {\em Optimal resource allocation in coordinated multi-cell systems},
\newblock Now Publishers Inc, 2013.

\bibitem{marzetta16}
T.~L. Marzetta, E.~G. Larsson, H.~Yang, and H.~Q. Ngo,
\newblock {\em Fundamentals of Massive {MIMO}},
\newblock Cambridge University Press, 2016.

\bibitem{ngo2017cell}
H.~Q. Ngo, A.~Ashikhmin, H.~Yang, E.~G. Larsson, and T.~L. Marzetta,
\newblock ``Cell-free massive {MIMO} versus small cells,''
\newblock {\em IEEE Transactions on Wireless Communications}, vol. 16, no. 3,
  pp. 1834--1850, 2017.

\bibitem{demir2021}
{\"O}.~T. Demir, E.~Bj{\"o}rnson, and L.~Sanguinetti,
\newblock ``Foundations of user-centric cell-free massive {MIMO},''
\newblock {\em Foundations and Trends in Signal Processing}, vol. 14, pp.
  162--472, 2021.

\bibitem{donato2021}
G.~Interdonato, H.~Q. Ngo, and E.~G. Larsson,
\newblock ``Enhanced normalized conjugate beamforming for cell-free massive
  {MIMO},''
\newblock {\em IEEE Transactions on Communications}, vol. 69, no. 5, pp.
  2863--2877, 2021.

\bibitem{buzzi2020}
S.~Buzzi, C.~D’Andrea, A.~Zappone, and C.~D’Elia,
\newblock ``User-centric {5G} cellular networks: Resource allocation and
  comparison with the cell-free massive {MIMO} approach,''
\newblock {\em IEEE Transactions on Wireless Communications}, vol. 19, no. 2,
  pp. 1250--1264, 2020.

\bibitem{chaves2020convergence}
R.~S. Chaves, E.~Cetin, M.~V.~S Lima, and W.~A. Martins,
\newblock ``On the convergence of max-min fairness power allocation in massive
  {MIMO} systems,''
\newblock {\em IEEE Communications Letters}, vol. 24, no. 12, pp. 2873--2877,
  2020.

\bibitem{yang2017massive}
H.~Yang and T.~L. Marzetta,
\newblock ``Massive {MIMO} with max-min power control in line-of-sight
  propagation environment,''
\newblock {\em IEEE Transactions on Communications}, vol. 65, no. 11, pp.
  4685--4693, 2017.

\bibitem{massivemimobook}
E.~Bj\"{o}rnson, J.~Hoydis, and L.~Sanguinetti,
\newblock ``Massive {MIMO} networks: {Spectral}, energy, and hardware
  efficiency,''
\newblock {\em Foundations and Trends{\textregistered} in Signal Processing},
  vol. 11, no. 3-4, pp. 154--655, 2017.

\bibitem{das1997closer}
I.~Das and J.~E. Dennis,
\newblock ``A closer look at drawbacks of minimizing weighted sums of
  objectives for {P}areto set generation in multicriteria optimization
  problems,''
\newblock {\em Structural optimization}, vol. 14, no. 1, pp. 63--69, 1997.

\bibitem{conceicao2022heuristics}
F.~Conceição, C.~H. Antunes, M.~Gomes, V.~Silva, and R.~Dinis,
\newblock ``Max-min fairness optimization in uplink cell-free massive {MIMO}
  using meta-heuristics,''
\newblock {\em IEEE Transactions on Communications}, pp. 1--1, 2022.

\bibitem{schubert2019}
M.~Schubert, R.~Bohnke, and W.~Xu,
\newblock ``Multi-connectivity beamforming for enhanced reliability and massive
  access,''
\newblock in {\em European Conf. on Networks and Communications}, 2019,
\newblock available on https://www.researchgate.net/.

\bibitem{renatomaxmin}
R.~L.~G. Cavalcante and S.~Sta\'nczak,
\newblock ``Fundamental properties of solutions to utility maximization
  problems,''
\newblock in {\em arXiv:1610.01988}, 2016.

\bibitem{renato2016maxmin}
R.~L.~G. Cavalcante, M.~Kasparick, and S.~Sta\'nczak,
\newblock ``Max-min utility optimization in load coupled interference
  networks,''
\newblock {\em {IEEE} Trans. Wireless Commun.}, vol. 16, no. 2, pp. 705--716,
  Feb. 2017.

\bibitem{ismayilov2019power}
R.~Ismayilov, B.~Holfeld, R.~L.~G. Cavalcante, and M.~Kaneko,
\newblock ``Power and beam optimization for uplink millimeter-wave hotspot
  communication systems,''
\newblock in {\em 2019 IEEE Wireless Communications and Networking Conference
  (WCNC)}. IEEE, 2019, pp. 1--8.

\bibitem{nuzman07}
C.~J. Nuzman,
\newblock ``Contraction approach to power control, with non-monotonic
  applications,''
\newblock in {\em IEEE GLOBECOM 2007-IEEE Global Telecommunications
  Conference}. IEEE, 2007, pp. 5283--5287.

\bibitem{renato2019}
R.~L.~G. Cavalcante, Q.~Liao, and S.~Sta{\'n}czak,
\newblock ``Connections between spectral properties of asymptotic mappings and
  solutions to wireless network problems,''
\newblock {\em IEEE Transactions on Signal Processing}, vol. 67, no. 10, pp.
  2747--2760, 2019.

\bibitem{cai2012maxmin}
D.~W.~H. Cai, C.~W. Tan, and S.~H. Low,
\newblock ``Optimal max-min fairness rate control in wireless networks:
  Perron-frobenius characterization and algorithms,''
\newblock in {\em 2012 Proceedings IEEE INFOCOM}, 2012, pp. 648--656.

\bibitem{yates95}
R.~D. Yates,
\newblock ``A framework for uplink power control in cellular radio systems,''
\newblock {\em {IEEE} J. Select. Areas Commun.}, vol. 13, no. 7, pp. pp.
  1341--1348, Sept. 1995.

\bibitem{renato2016}
R.~L.~G. Cavalcante, Y.~Shen, and S.~Sta{\'n}czak,
\newblock ``Elementary properties of positive concave mappings with
  applications to network planning and optimization,''
\newblock {\em {IEEE} Trans. Signal Processing}, vol. 64, no. 7, pp.
  1774--1873, April 2016.

\bibitem{krause2015}
U.~Krause,
\newblock {\em Positive dynamical systems in discrete time: theory, models, and
  applications}, vol.~62,
\newblock Walter de Gruyter GmbH \& Co KG, 2015.

\bibitem{nussbaum1986convexity}
R.~D. Nussbaum,
\newblock ``Convexity and log convexity for the spectral radius,''
\newblock {\em Linear Algebra and its Applications}, vol. 73, pp. 59--122,
  1986.

\bibitem{miretti2021team}
L.~Miretti, E.~Björnson, and D.~Gesbert,
\newblock ``Team {MMSE} precoding with applications to cell-free massive
  {MIMO},''
\newblock {\em {IEEE} Trans. Wireless Commun.}, 2022,
\newblock Early {A}ccess.

\bibitem{miretti2021team2}
L.~Miretti, E.~Björnson, and D.~Gesbert,
\newblock ``Team precoding towards scalable cell-free massive {MIMO}
  networks,''
\newblock {\em 2021 55th Asilomar Conference on Signals, Systems, and
  Computers}, 2021.

\bibitem{bjornson2019making}
E.~Bj{\"o}rnson and L.~Sanguinetti,
\newblock ``Making cell-free massive {MIMO} competitive with {MMSE} processing
  and centralized implementation,''
\newblock {\em IEEE Transactions on Wireless Communications}, vol. 19, no. 1,
  pp. 77--90, 2019.

\end{thebibliography}

\end{document}